\documentclass{llncs}
\usepackage{graphicx}
\usepackage{algorithm}
\usepackage{algorithmic}
\usepackage{amsmath}
\usepackage{amssymb}
\usepackage{comment}
\usepackage{url}
\usepackage{listings, tikz}
\usetikzlibrary{positioning,shapes}
\usepackage{verbatimbox}

\newtheorem{Theorem}{Theorem}

\newcommand{\keyword}[1]{\mathsf{#1}}
\newcommand{\kw}[1]{\keyword{#1}}

\newcommand{\textcode}[1]{\lstinline|#1|}
\newcommand{\tc}[1]{\lstinline|#1|}

\newcommand{\kwnull}[0]{\keyword{null}}
\newcommand{\kwnew}[0]{\keyword{new}}
\newcommand{\kwextends}[0]{\keyword{extends}}
\newcommand{\kwclass}[0]{\keyword{class}}

\newcommand\Var{\mathtt{VAR}}

\newcommand\Obj{\mathtt{OBJ}}

\newcommand{\VPT}{\Omega}
\newcommand{\HPT}{\Phi}
\newcommand{\Class}{\mathcal{C}}
\newcommand{\Field}{\mathcal{F}}

\newcommand{\less}{\sqsubseteq}
\newcommand{\tflow}{\dashrightarrow}
\newcommand{\hflow}{\longrightarrow}
\newcommand{\lhflow}[1]{\stackrel{#1}{\hflow}}

\newcommand\set[1]{\{#1\}}
\newcommand\power{\mathcal{P}}

\begin{document}
\title{A Relational Static Semantics for Call Graph Construction}
%
%
\author{Xilong Zhuo \and Chenyi Zhang}
%
%
%
\institute{Jinan University}

\maketitle              
\begin{abstract}
The problem of resolving virtual method and interface calls in object-oriented languages has been a long standing challenge to the program analysis community. The complexities are due to various reasons, such as increased levels of class inheritance and polymorphism in large programs. In this paper, we propose a new approach called type flow analysis that represent propagation of type information between program variables by a group of relations without the help of a heap abstraction. We prove that regarding the precision on reachability of class information to a variable, our method produces results equivalent to that one can derive from a points-to analysis. Moreover, in practice, our method consumes lower time and space usage, as supported by the experimental results.

\end{abstract}

\section{Introduction}\label{sec:introduction}

For object-oriented programming languages, virtual methods (or functions) are those declared in a base class but are meant to be overridden in different child classes. Statically determine a set of methods that may be invoked at a call site is important to program optimization, from results of which a subsequent optimization may reduce the cost of virtual function calls or perform method inlining if target method forms a singleton set, and one may also remove methods that are never called by any call sites, or produce a call graph which can be useful in other optimization~processes.

Efficient solutions, such as Class Hierarchy Analysis (CHA)~\cite{Dean1995,Fernandez1995}, Rapid Type Analysis (RTA)~\cite{Bacon1996} and Variable Type Analysis (VTA)~\cite{Sundaresan2000}, conservatively assign each variable a set of class definitions, with relatively low precision. Alternatively, with the help of an abstract heap, one may take advantage of points-to analysis~\cite{andersen94} to compute a set of object abstractions that a variable may refer to, and resolve the receiver classes in order to find associated methods at call sites.

The algorithms used by CHA, RTA and VTA are conservative, which aim to provide an efficient way to resolve calling edges, and which usually take linear-time in the size of a program, by focusing on the types that are collected at the receiver of a call site. For instance, let $x$ be a variable of declared class $A$, then at a call site $x.m()$, CHA will draw a call edge from this call site to method $m()$ of class $A$ and every definition $m()$ of a class that extends $A$. In case class $A$ does not define $m()$, a call edge to an ancestor class that defines $m()$ will also be included. For a variable $x$ of declared interface $I$, CHA will draw a call edge from this call site to every method of name $m()$ defined in class $X$ that implements $I$.
We write $CHA(x,m)$ for the set of methods that are connected from call site $x.m()$ as resolved by  Class Hierarchy Analysis (CHA).
Rapid Type Analysis (RTA) is an improvement from CHA which resolves call site $x.m()$ to $CHA(x,m)\cap inst(P)$, where $inst(P)$ stands for the set of methods of classes that are instantiated in the program.

\begin{figure}[t!]
\begin{minipage}[t]{0.5\linewidth}
\centering
\begin{verbbox}
class A{
    A f;
    void m(){
        return this.f;
    }
}

class B extends A{}

class C extends A{}
\end{verbbox}
\theverbbox
\end{minipage}
\begin{minipage}[t]{0.5\linewidth}
\centering
\begin{verbbox}
1:  A x = new A();  //O_1
2:  B b = new B();  //O_2
3:  A y = new A();  //O_3
4:  C c = new C();  //O_4
5:  x.f = b;
6:  y.f = c;
7:  z = x.m();
\end{verbbox}
\theverbbox
\end{minipage}
\caption{An example that compares precision on type flow in a program.}\label{figure:example}
\end{figure}

Variable Type Analysis (VTA) is a further improvement. VTA defines a node for each variable, method, method parameter and field. Class names are treated as values and propagation of such values between variables work in the way of value flow.
As shown in Figure~\ref{figure:vta}, the statements on line $1-4$ initialize type information for variables $x$, $y$, $b$ and $c$,
and statements on line $5-7$ establish value flow relations. Since both $x$ and $y$ are assigned type $A$, $x.f$ and $y.f$ are both represented by node $A.f$, thus the set of types reaching $A.f$ is now $\set{B,C}$. (Note this is a more precise result than CHA and RTA which assign $A.f$ with the set $\set{A,B,C}$.)
Since $A.m.this$ refers to $x$, $this.f$ inside method $A.m()$ now refers to $A.f$. Therefore, through $A.m.return$, $z$ receives  $\set{B,C}$ as its final set of reaching types.

\begin{figure}[t!]
\begin{tabular}{cc}

\begin{minipage}[h]{0.5\linewidth}
\centering
\begin{tabular}{|c|c|}
  \hline
  \textbf{Statement} & \textbf{VTA fact} \\
   \hline
  \hline
  $A\ x = \kwnew\ A()$ & $x\leftarrow A$ \\ \hline
  $B\ b = \kwnew\ B()$ & $b\leftarrow B$ \\ \hline
  $A\ y = \kwnew\ A()$ & $y\leftarrow A$ \\ \hline
  $C\ c = \kwnew\ C()$ & $c\leftarrow C$ \\ \hline
  $x.f = b$ & $A.f\leftarrow b$ \\ \hline
  $y.f = c$ & $A.f\leftarrow c$ \\ \hline
   & $A.m.this\leftarrow x$ \\
  $z = x.m()$  & $A.m.return\leftarrow A.f$ \\
   & $z\leftarrow A.m.return$\\
  \hline
\end{tabular}
\caption{VTA facts on the example}\label{figure:vta}
\end{minipage}
&
\begin{minipage}[h]{0.5\linewidth}
\centering
\includegraphics[width=3.5cm]{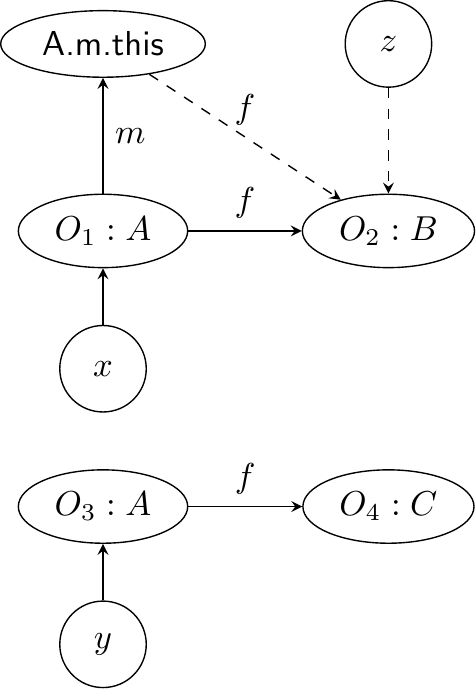}
\caption{Points-to results on the example}\label{figure:points-to}
\end{minipage}
\end{tabular}
\end{figure}

The result of a context-insensitive subset based points-to analysis~\cite{andersen94} creates a heap abstraction of four objects (shown on line $1-4$ of Figure~\ref{figure:example} as well as in the ellipses in Figure~\ref{figure:points-to}). These abstract objects are then inter-connected via field store access defined on line $5-6$. The derived field access from $A.m.this$ to $O_2$ is shown in dashed arrow. By return of the method call $z=x.m()$, variable $z$ receives $O_2$ of type $B$ from $A.m.this.f$, which gives a more precise set of reaching types for variable $z$.

From this example, one may conclude that the imprecision of VTA in comparison with points-to analysis is due to the over abstraction of object types, such that $O_1$ and $O_3$, both of type $A$, are treated under the same type. Nevertheless, points-to analysis requires to construct a heap abstraction, which brings in extra information, especially when we are only interested in the set of reaching types of a variable.

\begin{figure}[t!]
\centering
\includegraphics[width=8cm]{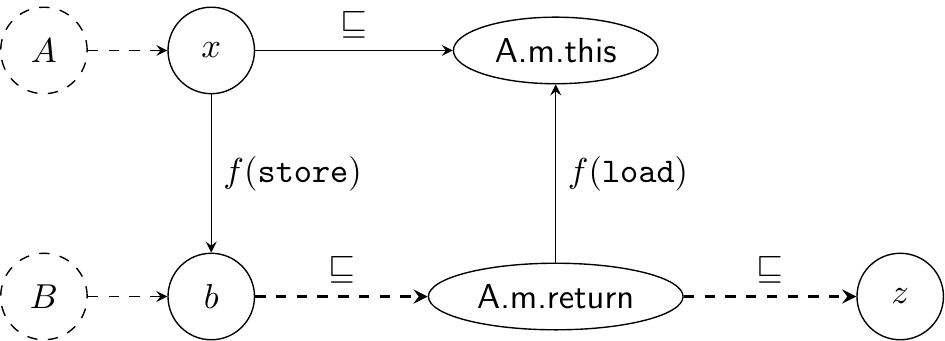}
\caption{Type Flow Analysis for variable $z$ in the example}\label{fig:tfa}
\end{figure}

In this paper we introduce a relational static semantics called Type Flow Analysis (TFA) on program variables and field accesses. Different from VTA, besides a binary value flow relation $\less$ on the variable domain $\Var$, where $x\less y$ denotes all types that flow to $x$ also flow to $y$, we also build a ternary field store relation $\rightarrow\ \subseteq\Var\times\Field\times\Var$ to trace the \emph{load} and \emph{store} relationship between variables via field accesses. This provides us additional ways to extend the relations $\less$ as well as $\rightarrow$. 
Taking the example from Figure~\ref{figure:example}, we are able to collect the store relation $x\lhflow{f}b$ from line $5$. Since $x\less \textsf{A.m.this}$, together with the implicit assignment which loads $f$ of $\textsf{A.m.return}$, we further derives $b\less\textsf{A.m.return}$ and $b\less z$ (dashed thick arrows in Figure~\ref{fig:tfa}). Therefore, we assign type $B$ to variable $z$. The complete reasoning pattern is depicted in Figure~\ref{fig:tfa}. Nevertheless, one cannot derive $c\less z$ in the same way.

We have proved that in the context-insensitive inter-procedural setting, TFA is as precise as the subset based points-to analysis regarding type related information. Since points-to analysis can be enhanced with various types of context-sensitivity on variables and objects (e.g., call-site-sensitivity~\cite{Shivers91,Kastrinis2013}, object-sensitivity~\cite{Milanova2005,Smaragdakis11,Tan16} and type-sensitivity~\cite{Smaragdakis11}), a context-sensitive type flow analysis will only require to consider contexts on variables, which is left for future work. The context-insensitive type flow analysis has been implemented in the Soot framework~\cite{soot}, and the implementation has been tested on a collection of benchmark programs from SPECjvm2008~\cite{specjvm} and DaCapo~\cite{Blackburn2006}. 
The initial experimental result has shown that TFA consumes similar or less runtime than CHA~\cite{Dean1995}, but has precision comparable to that of a points-to analysis.

\section{Type Flow Analysis}\label{sec:type-flow-analysis}

We define a core calculus consisting of most of the key object-oriented language features, as shown in Figure~\ref{fig:syntax}.
A program is defined as a code base $\overline{C}$ (i.e., a collection of class definitions) with statement $s$ to be evaluated.
To run a program, one may assume that $s$ is the default (static) entry method with local variable declarations $\overline{D}$,
similar to e.g., Java and C++,
which may differ in specific language designs.
We define a few auxiliary functions. $fields$ maps class names to its fields, $methods$ maps class names to its defined or inherited methods, and $type$ provides types (or class names) for objects. Given class $c$, if $f\in fields(c)$, then $ftype(c,f)$ is the defined class type of field $f$ in $c$. Similarly, give an object $o$, if $f\in fields(type(o))$, then $o.f$ may refer to an object of type $ftype(type(o),f)$,
or an object of any of its subclass at runtime. Write $\Class$ for the set of classes, $\Obj$ for the set of objects, $\Field$ for the set of fields and $\Var$ for the set of variables that appear in a program.


\begin{figure}[!htbp]\centering
	\begin{tabular}[c]{lll} 
		$C$&$::=$&$\kwclass\ c\ [\kwextends\ c] \ \{\overline{F};\ \overline{M}\}$\\
        $F$&$::=$&$c \ f$\\
        $D$&$::=$&$c \ z$\\
		$M$&$::=$&$m(x) \ \{\overline{D}; s; \kw{return}\ x'\}$\\
		$s$&$::=$&$e\mid x{=}\kwnew \ c\mid  x {=} e \mid x.f{=}y \mid s;s$\\
		$e$&$::=$&$ \kwnull\mid x \mid x.f \mid x.m(y) $\\
        $prog$&$::=$&$\overline{C};\overline{D}; s$\\
	\end{tabular}
	\caption{Abstract syntax for the core language. \label{fig:syntax}}
\end{figure}

In this simple language we do not model common types (e.g., $\kw{int}$ and $\kw{float}$) that are irrelevant to our analysis, and we focus on the reference types which form a class hierarchical structure. Similar to Variable Type Analysis (VTA), we assume a context insensitive setting, such that every variable can be uniquely determined by its name together with its enclosing class and methods.
For example, if a local variable $x$ is defined in method $m$ of class $c$, then $c.m.x$ is the unique representation of that variable.
Therefore, it is safe to drop the enclosing the class and method name if it is clear from the context.
In general, we have the following types of variables in our analysis: (1) local variables, (2) method parameters, (3) this reference of each method, all of which are syntactically bounded by their enclosing methods and classes.

We enrich the variable type analysis with the new type flow analysis by using three relations,
a partial order on variables $\less\ \subseteq\Var\times\Var$, a type flow relation
$\tflow\subseteq\Class\times\Var$, as well as a field access relation $\hflow\subseteq\Var\times\Field\times\Var$,
which are initially given as follows.
\begin{definition}\label{def:base} (Base Relations)
We have the following base facts for the three relations.
\begin{enumerate}
  \item $c\tflow x$ if there is a statement $x = \kwnew\ c$;
  \item $y\less x$ if there is a statement $x = y $;
  \item $x\lhflow{f}y$ if there is a statement $x.f = y$.  
\end{enumerate}
\end{definition}
Intuitively, $c\tflow x$ means variable $x$ may have type $c$ (i.e., $c$ flows to $x$), $y\less x$ means all types flow to $y$ also flow to $x$, and $x\lhflow{f}y$ means from variable $x$ and field $f$ one may access variable $y$.\footnote{Note that VTA treats statement $x.f = y$ as follows. For each class $c$ that flows to $x$ which defines field $f$, VTA assigns all types that flow to $y$ also to $c.f$.} These three relations are then extended by the following rules.
\begin{definition}\label{def:extension} (Extended Relations)
\begin{enumerate}
  \item For all statements $x = y.f $, if $y\lhflow{f\ *}z$, then $z\less^* x$.
  \item $c\tflow^* y$ if $c\tflow y$, or $\exists x\in\Var:c\tflow^* x\wedge x\less^* y$;
  \item $y\less^* x$ if $x=y$ or $y\less x$ or $\exists z\in\Var:y\less^* z\wedge z\less^* x$;
  \item $y\lhflow{f\ *}z$ if $\exists x\in\Var: x\lhflow{f}z\wedge (\exists z'\in\Var: z'\less^* y \wedge z'\less^*x)$;
  \item The type information is used to resolve each method call $x = y.m(z)$.
  \begin{equation*}
  \left.\begin{array}{l}\forall c\tflow^* y:\\ m(z')\{\dots \kw{return}\ x'\}\in methods(c):\end{array}\right.\left\{\begin{array}{l}
        z\less^* c.m.z'\\
        c\tflow^* c.m.this\\
        c,m.x'\less^* x \\
        \end{array}\right.
  \end{equation*}
\end{enumerate}
\end{definition}

The final relations are the least relations that satisfy constraints of Definition~\ref{def:extension}.
Comparing to VTA~\cite{Sundaresan2000}, we do not have field reference $c.f$ for each class $c$ defined in a program. Instead, we define a relation that connects the two variable names and one field name. 
Although the three relations are inter-dependent, one may find that without method call (i.e., Definition~\ref{def:extension}.5), a smallest model satisfying the two relations $\rightarrow^*$ (field access) and $\less^*$ (variable partial order) can be uniquely determined without considering the type flow relation $\tflow^*$.

In order to compare the precision of TFA with points-to analysis, we present a brief list of the classic subset-based points-to rules for our language in Figure~\ref{fig:constraints}, where $param(type(o),m))$, $this(type(o),m))$ and $return(type(o),m)$ refer to the formal parameter, \textsf{this} reference and \textsf{return} variable of the method $m$ of the class for which object $o$ is declared, respectively.

\begin{figure*}
	\centering %
    \begin{tabular}{|l|c|}
        \hline
    \textbf{statement} & \textbf{Points-to constraints} \\
    \hline
    $x = \kwnew\ c$ & $o_i\in\VPT(x)$\\
    \hline
    $x = y $ & $\VPT(y)\subseteq\VPT(x)$\\
    \hline
    $x = y.f $ & $\forall o\in\VPT(y):\HPT(o,f)\subseteq\VPT(x)$\\
    \hline
    $x.f = y $ & $\forall o\in\VPT(x):\VPT(y)\subseteq\HPT(o,f)$\\
    \hline
    $x=y.m(z)$ &
        \(\forall o\in\VPT(y):\left\{\begin{array}{l}
        \VPT(z)\subseteq\VPT(param(type(o),m))\\
        \VPT(this(type(o),m))=\set{o}\\
        \forall x'\in return(type(o),m):\\ \hspace{35pt} \VPT(x')\subseteq\VPT(x) \end{array}\right.\)
        \\
    \hline
	\end{tabular}
\caption{Constraints for points-to analysis. \label{fig:constraints}}
\end{figure*}

To this end we present the first result of the paper, which basically says type flow analysis has the same precision regarding type based check, such as call site resolution and cast failure check, when comparing with the points-to analysis.
\begin{Theorem}~\label{thm:tfa}
  In a context-insensitive analysis, for all variables $x$ and classes $c$, $c\tflow^*x$ iff there exists an object abstraction $o$ of $c$ such that $o\in\VPT(x)$.
\end{Theorem}
\begin{proof} (sketch)
For a proof sketch, first we assume every object creation site $x = \kwnew\ c_i$ at line $i$ defines a mini-type $c_i$, and if the theorem is satisfied in this setting, a subsequent merging of mini-types into classes will preserve the result.

Moreover, we only need to prove the intraprocedural setting which is the result of Lemma~\ref{lem:tfa-intra}. Because if in the intraprocedural setting the two systems have the same smallest model for all methods, then at each call site $x=y.m(a)$ both analyses will assign $y$ the same set of classes and thus resolve the call site to the same set of method definitions, and as a consequence, each method body will be given the same set of extra conditions, thus all methods will have the same initial condition for the next round iteration. Therefore, both inter-procedural systems will eventually stabilize at the same model. \qed
\end{proof}


\begin{lemma}\label{lem:tfa-intra}
In a context-insensitive intraprocedural analysis where each class $c$ only syntactically appears once in the form of $\kwnew\ c$, for all variables $x$ and classes $c$, $c\tflow^*x$ iff there exists an object abstraction $o$ of type $c$ such that $o\in\VPT(x)$.
\end{lemma}
\begin{proof}
Since the points-to constraints define the smallest model $(\VPT, \HPT)$ with $\VPT:\Var\rightarrow\Obj$ and $\HPT:\Obj\times\Field\rightarrow\power(\Obj)$, and the three relations of type flow analysis also define the smallest model that satisfies Definition~\ref{def:base} and Definition~\ref{def:extension}, we prove that every model of points-to constraints is also a model of TFA, and vice versa. Then the least model of both systems must be the same, as otherwise it would lead to contradiction.

\medskip

($\star$) For the `only if' part ($\Rightarrow$), we define $Reaches(x)=\set{c\mid c\tflow^* x}$, and assume a bijection $\xi:\Class\rightarrow\Obj$ that maps each class $c$ to the unique object $o$ that is defined (and $type(o)=c$). Then we construct a function $Access:\Class\times\Field\rightarrow\power(\Class)$ and show that $(\xi(Reaches),\xi(Access))$ satisfies the points-to constraints. Define $Access(c,f)=\set{c'\mid x\lhflow{f\ *}y\wedge c\in Reaches(x)\wedge c'\in Reaches(y)}$. We prove the following cases according to the top four points-to constraints in Figure~\ref{fig:constraints}.
\begin{itemize}
\item For each statement $x = \kwnew\ c$, we have $\xi(c)\in\xi(Reaches(x))$;
\item For each statement $x = y$, we have $Reaches(y)\subseteq Reaches(x)$ and $\xi(Reaches(y))\subseteq\xi(Reaches(x))$;
\item For each statement $x.f = y$, we have $x\lhflow{f}y$, then by definition for all $c\in Reaches(x)$, and $c'\in Reaches(y)$, we have $c'\in Access(c,f)$, therefore $\xi(c')\in\xi(Reaches(y))$ we have $\xi(c')\in \xi(Access(\xi(c),f))$.
\item For each statement $x = y.f$, let $c\in Reaches(y)$, we need to show $\xi(Access(c,f))\subseteq\xi(Reaches(x))$, or equivalently, $Access(c,f)\subseteq Access(x)$. Let $c'\in Access(c,f)$, then by definition, there exist $z$, $z'$ such that $c\in Reaches(z)$, $c'\in Reaches(z')$ and $z\lhflow{f\ *}z'$. By $c\in Reaches(y)$ and Definition~\ref{def:extension}.4, we have $y\lhflow{f\ *}z'$. Then by Definition~\ref{def:extension}.1, $z'\less^*x$. Therefore $c'\in Reaches(x)$.
\end{itemize}

\medskip

($\star$) For the `if' part ($\Leftarrow$), let ($\VPT$, $\HPT$) be a model that satisfies all the top four constraints defined in Figure~\ref{fig:constraints}, and a bijection $\xi:\Class\rightarrow\Obj$, we show the following constructed relations satisfy value points-to.
\begin{itemize}
  \item For all types $c$ and variables $x$, $c\tflow^* x$ if $\xi(c)\in\VPT(x)$;
  \item For all variables $x$ and $y$, $x\less^*y$ if $\VPT(x)\subseteq\VPT(y)$;
  \item For all variables $x$ and $y$, and for all fields $f$, $x\lhflow{f\ *}y$ if for all $o_1,o_2\in\Obj$ such that $o_1\in\VPT(x)$ and $o_2\in\VPT(y)$ then $o_2\in\HPT(o_1,f)$.
\end{itemize}
We check the following cases for the three relations $\tflow^*$, $\less^*$ and $\rightarrow$ that are just defined from the above.
\begin{itemize}
\item For each statement $x = \kwnew\ c$, we have $\xi(c)\in\VPT(x)$, so $c\tflow^* x$ by definition.
\item For each statement $x = y$, we have $\VPT(y)\subseteq\VPT(x)$, therefore $y\less^*x$ by definition.
\item For each statement $x.f = y$, we have for all $o_1\in\VPT(x)$ and $o_2\in\VPT(y)$, $o_2\in\HPT(o_1,f)$, which derives $x\lhflow{f\ *}y$ by definition.
\item  For each statement $x = y.f$, given $y\lhflow{f\ *}z$, we need to show $z\less^* x$. Equivalently, by definition we have for all $o_1\in\VPT(y)$ and $o_2\in\VPT(z)$, $o_2\in\HPT(o_1,f)$. Since points-to relation gives $\HPT(o_1,f)\subseteq\VPT(x)$, we have $o_2\in\VPT(x)$, which derives $\VPT(z)\subseteq\VPT(x)$, the definition of $z\less^* x$.
\item The proof for the properties in the rest of Definition~\ref{def:extension} are related to transitivity of the three TFA relations, which are straightforward. We leave them for interested readers. \qed
\end{itemize}
\end{proof}


\section{Implementation and Optimization}\label{sec:minimization}

The analysis algorithm is written in Java, and is implemented in the Soot framework~\cite{soot}, the most popular static analysis framework for Java. The three base relations (i.e., $\tflow$, $\less$ and $\rightarrow$) of Definition~\ref{def:base} are extracted from Soot's intermediate representation and the extended relations (i.e., $\tflow^*$, $\less^*$ and $\rightarrow^*$) of Definition~\ref{def:extension} are then computed considering the mutual dependency relations between them. Since we are only interested in reference types, we do not carry out analysis on basic types such as \textsf{boolean}, \textsf{int} and \textsf{double}. We also do not consider more advanced Java features such as functional interfaces and lambda expressions, as well as usages of Java Native Interface (JNI), nor method calls via Java reflective API. We have not tried to apply the approach to Java libraries, all invocation of methods from JDK are treated as end points, thus all possible call back edges will be missed in the analysis.
Array accesses are treated conservatively---all type information that flows to one member of a reference array flows to all members of that array, so that only one node is generated for each~array.

Since call graph information may be saved and be used for subsequent analyses, we propose the following two ways to reduce storage for computed result. If a number of variables are similar regarding type information in a graph representation, they can be merged and then referred to by the merged node.

\begin{enumerate}
  \item If $x\less^*y$ and $y\less^*x$, we say $x$ and $y$ form an \emph{alias pair}, written $x\sim y$. Intuitively, $\Var/\sim$ is a partition of $\Var$ such that each $c\in\Var/\sim$ is a strongly connected component (SCC) in the variable graph edged by relation $\less$, which can be quickly collected by using Tarjan's algorithm~\cite{Tarjan72}.
  \item A more aggressive compression can be achieved in a way similar to bisimulation minimization of finite state systems~\cite{Kanellakis90,Paige87}. Define $\approx\ \subset\Var\times\Var$ such that $x\approx y$ is symmetric and if
      \begin{itemize}
      \item for all class $c$, $c\tflow^* x$ iff $c\tflow^*y$, and
      \item for all $x\lhflow{f\ *}x'$ there exists $y\lhflow{f\ *}y'$ and $x'\approx y'$.
      \end{itemize}
\end{enumerate}

It is straightforward to see that $\approx$ is a more aggressive merging scheme.
\begin{lemma}\label{lem:bisimilar}
For all $x,y\in\Var$, $x\sim y$ implies $x\approx y$.
\end{lemma}


We have implemented the second storage minimization scheme by using Kanellakis and Smolka's algorithm~\cite{Kanellakis90} which computes the largest bisimulation relation for a given finite state labelled transition system. In our interpretation, the variables are treated as states and the field access relation is treated as the state transition relation. The algorithm then merges equivalent variables into a single group. As a storage optimization process, this implementation has been tested and evaluated in the next section.

\section{Experiment and Evaluation}\label{sec:experiment}

We evaluate our approach by measuring its performance on $13$ benchmark programs. Among the benchmark programs, \textit{compress}, \textit{crypto} are from the SPECjvm2008 suite~\cite{specjvm}, and the other $11$ programs are from the DaCapo suite~\cite{Blackburn2006}. We randomly selected these test cases from the two benchmark suites, in order to test code bases that are representative from a variety of sizes.
All of our experiments were conducted on a Huawei Laptop equipped with Intel i5-8250U processor at 1.60 GHz and 8 GB memory, running Ubuntu 16.04LTS with OpenJDK 1.8.0.

We compare our approach against the default implementation of Class Hierarchy Analysis (CHA) and context-insensitive points-to analysis~\cite{Lhotak2003} that are implemented by the Soot team. We use Soot as our basic framework to extract the SSA based representation of the benchmark code. We also generate automata representation for the resulting relations which can be visualized in a subsequent user-friendly manual inspection. The choice of the context-insensitive points-to analysis is due to our approach also being context-insensitive, thus the results will be comparable. In the following tables we use CHA, PTA and TFA to refer to the results related to class hierarchy analysis, points-to analysis and type flow analysis, respectively.

During the evaluation the following three research questions are addressed.

\paragraph{RQ 1:} How efficient is our approach compared with the traditional class hierarchy analysis and points-to analysis?

\paragraph{RQ 2:} How accurate is the result of our approach when comparing with the other analyses?

\paragraph{RQ 3:} Does our optimization (or minimization) algorithm achieve significantly reduce storage consumption?


\begin{figure}[t!]\centering
\begin{tabular}{lcccccc}
	\hline
	\textbf{bench} & \textbf{T$_{\textsf{CHA}}(s)$} & \textbf{T$_{\textsf{PTA}}(s)$} & \textbf{T$_{\textsf{TFA}}(s)$} & \hspace{2pt}\textbf{R${\tflow}$}\hspace{2pt} & \hspace{2pt}\textbf{R$_{\less}$}\hspace{2pt} & \hspace{2pt}\textbf{R${\rightarrow}$} \hspace{2pt}\\
	\hline
	compress & 0.02 & 0.12 & 0.02 & 87 & 202 & 24\\
	crypto & 0.01 & 0.12 & 0.03 & 94 & 226 & 18\\
	bootstrap & 23.74 & 34.13 & 0.03 & 191 & 453 & 17\\
	commons-codec & 0.009 & 0.12 & 0.13 & 306 & 3324 & 49\\
	junit & 24.45 & 34.24 & 0.18 & 1075 & 5772 & 241\\
	commons-httpclient & 0.009 & 0.12 & 0.36 & 2423 & 8511 & 521 \\
	serializer & 22.60 & 32.68 & 1.71 & 3006 & 17726 & 331\\
	xerces & 21.69 & 34.83 & 3.74 & 12590 & 72503 & 2779\\
	eclipse & 21.95 & 42.64 & 1.68 & 7933 & 37435 & 1620\\
	derby & 22.77 & 48.82 & 18.09 & 20698 & 191854 & 5386\\
	xalan & 78.40 & - & 42.11 & 32971 & 162249 & 3696\\
	antlr & 44.20 & - & 3.96 & 16117 & 76741 & 3879\\
	batik & 45.84 & - & 6.38 & 29409 & 122534 & 6039 \\
	\hline
\end{tabular}
\caption{Runtime cost with different analysis}
\label{experiment:TimeCost}
\end{figure}

\subsection{RQ1: Efficiency}\label{subsec:efficiency}
To answer the first research question, we executed each benchmark program $10$ times with the CHA, PTA and TFA algorithms. We calculated the average time consumption as displayed at the left three columns of the Table in Figure~\ref{experiment:TimeCost}. The sizes of each relation that our approach generated (i.e., the type flow relation `$\tflow$', variable partial order `$\less$' and the field access `$\rightarrow$') are counted, which provides an estimation of size for the problem we are treating. One may observe that when the problem size increases, the execution time of the our algorithm also increases in a way similar to CHA, though in general the runtime of CHA is supposed to grow linearly in the size of a program. The reason that TFA sometimes outperforms CHA may be partially due to the size of the intrinsic complexity of the class and interface hierarchical structure that a program adopts. TFA is in general more efficient than the points-to analysis.
The runtime cost in TFA basically depends on the size of generated relations, as well as the relational complexity as most of the time is consumed to calculate a fixpoint. For PTA it also requires extra time for maintaining and updating a heap abstraction.

Taking a closer look at the benchmark \textit{bootstrap}, CHA and PTA analyze the benchmark using about $23.74$ and $34.13$ seconds, respectively. As TFA only generated $661$ relations, the analysis only takes $0.03$ second. 
There are $3$ benchmarks, \textit{xalan, antlr, batik}, that cannot be analyzed by PTA, marked as ``\emph{-}'' in the table in Figure~\ref{experiment:TimeCost}. In these cases the points-to analysis has caused an exception called ``OutOfMemoryError with JVM GC overhead limit exceeded'', as the JVM garbage collector is taking an excessive amount of time (by default $98\%$ of all CPU time of the process) and recovers very little memory in each run (by default $2\%$ of the heap).

\begin{figure}[t!]
  \centering
\begin{tabular}{lcccc}
	\hline
	\textbf{bench} & \textbf{CS$_{\textsf{base}}$} &\hspace{5pt} \textbf{CS$_{\textsf{CHA}}$}\hspace{5pt} &\hspace{5pt} \textbf{CS$_{\textsf{PTA}}$}\hspace{5pt} & \hspace{5pt} \textbf{CS$_{\textsf{TFA}}$}\hspace{5pt} \\
	\hline
	compress & 153 & 160 & 18 & 73 \\
	crypto & 302 & 307 & 62 & 121 \\
	bootstrap & 657 & 801 & 891 & 328 \\
	commons-codec & 1162 & 1372 & 270 & 554 \\
	junit & 3196 & 17532 & 11176 & 1197 \\
	commons-httpclient & 6817 & 17118 & 567 & 2928 \\
	serializer & 4782 & 9533 & 1248 & 1756 \\
	xerces & 24579 & 56252 & 10631 & 8120 \\
	eclipse & 23607 & 95073 & 70016 & 9379 \\
	derby & 69537 & 180428 & 85212 & 16381 \\
	xalan & 57430 & 155866 & - & 18669 \\
	antlr & 62007 & 147014 & - & 17177 \\
	batik & 56877 & 235071 & - & 20901 \\
	\hline
\end{tabular}
\caption{Call sites generated by different analyses}
\label{experiment:Callsite}
\end{figure}

\subsection{RQ2: Accurancy}\label{subsec:accurancy}

We answer the second question by considering the number of generated call sites as an indication of accuracy. In type flow analysis, a method call $a.m()$ is resolved to $c.m()$ if class $c$ is included in $a$'s reaching type set and method $m()$ is defined for $c$. In general, a more accurate analysis often generates a smaller set of types for each calling variable, resulting fewer call edges in total in the call graph. The table included in Figure~\ref{experiment:Callsite} displays the number of call sites generated by different analyses. We also include the base call site counting, i.e., the number of call sites syntactically written in the source code, as the baseline at the \textbf{CS$_{\textsf{base}}$} column. In comparison to CHA, our approach has reduced a significant amount of call site edges.
Comparing to other two analyses, the number of call edges resolved by TFA are often larger than PTA and smaller than CHA on the same benchmark. The difference may be caused by our over approximation on analyzing array references, as well as the existence of unsolved call edges from e.g. JNI calls or reflective calls.


\begin{figure}[t!]
  \centering
\begin{tabular}{lcccc}
	\hline
	\textbf{bench} & \textbf{Node$_{\textsf{origin}}$} & \textbf{Node$_{\textsf{opt}}$} & \textbf{Reduce} & \hspace{2pt}\textbf{Time$(\textsf{s})$} \\
	\hline
	compress & 205 & 130 & $36.59\%$ & 0.008 \\
	crypto & 312 & 158 & $49.36\%$ & 0.010 \\
	bootstrap & 514 & 279 & $45.72\%$ & 0.019 \\
	commons-codec & 1742 & 886 & $49.14\%$ & 0.202 \\
	junit & 5859 & 3243 & $44.65\%$ & 2.269 \\
	commons-httpclient & 9708 & 5164 & $46.81\%$ & 4.651 \\
	serializer & 9600 & 6668 & $30.54\%$ & 3.198 \\
	xerces & 41634 & - & - & -\\
	eclipse & 34631 & - & - & -\\
	derby & 112265 & - & - & -\\
	xalan & 103697 & - & - & -\\
	antlr & 56589 & - & - & -\\
	batik & 89336 & - & - & -\\
	\hline
\end{tabular}
\caption{Optimization result}
\label{experiment:Optimalization}
\end{figure}

\subsection{RQ3: Optimization}\label{subsec:optimization}

We apply bisimulation minimization to merge nodes that have the same types as well as accessible types through fields. Thus we can reduce the space consumption when storing the result for subsequent analysis processes. Regarding the third research question, we calculated the number of ``effective'' nodes before and after optimization process. Besides, time consumption is another factor that we consider. The results are shown in the table in Figure~\ref{experiment:Optimalization}. We evaluated our optimization algorithm on all benchmarks 10 times and successfully executed $7$ out of the $13$ benchmarks. For those succeeded, the algorithm has reduced space usage by about $45\%$ on average, with a reasonable time consumption. The other $6$ benchmarks cannot be due to the exception ``OutOfMemoryError with JVM GC overhead limit exceeded'' being thrown, for which we leave the symbol ``-'' accordingly in the table.

\section{Related Work}\label{sec:related-work}

There are not many works focusing on general purpose call graph construction algorithms, and we give a brief review of these works first.

As stated in the introduction, Class Hierarchy Analysis (CHA)~\cite{Dean1995,Fernandez1995}, Rapid Type Analysis (RTA)~\cite{Bacon1996} and Variable Type Analysis (VTA)~\cite{Sundaresan2000} are efficient algorithms that conservatively resolves call sites without any help from points-to analysis. Grove et al.~\cite{Grove1997} introduced an approach to model context-sensitive and context-insensitive call graph construction. They define call graph in terms of three algorithm-specific parameter partial orders, and provide a method called Monotonic Refinement, potentially adding to the class sets of local variables and adding new contours to call sites, load sites, and store sites. Tip and Palsberg~\cite{Tip2000} Proposed four propagation-based call graph construction algorithms, CTA, MTA, FTA and XTA. CTA uses distinct sets for classes, MTA uses distinct sets for classes and fields, FTA uses distinct sets for classes and methods, and XTA uses distinct
sets for classes, fields, and methods. The constructed call graphs tend to contain slightly fewer method definitions when compared to RTA. It has been shown that associating a distinct set of types with each method in a class has a significantly greater impact on precision than using a distinct set for each field in a class. Reif et al.~\cite{Reif2016} study the construction of call graphs for Java libraries that abstract over all potential library usages, in a so-called \emph{open world} approach. They invented two concrete call graph algorithms for libraries based on adaptations of the CHA algorithm, to be used for software quality and security issues. In general they are interested in analyzing library without knowing client application, which is complementary to our work that has focus on client program while treating library calls as end nodes.

Call graphs may serve as a basis for points-to analysis, but often a points-to analysis implicitly computes a call graph on-the-fly, such as the context insensitive points-to algorithm implemented in Soot using SPARK~\cite{Lhotak2003}. Most context-sensitive points-to analysis algorithms (e.g.,~\cite{Milanova2005,Sridharan2006,Smaragdakis11,Tan16}) progress call edges together with value flow, to our knowledge. The main distinction of our approach from these points-to analysis is the usage of an abstract heap, as we are only interested in the actual reaching types of the receiver of a call. Nevertheless, unlike CHA and VTA, our methodology can be extended to context-sensitive settings.

\section{Conclusion}\label{sec:conclusion}

In this paper we have proposed Type Flow Analysis (TFA), an algorithm that constructs call graph edges for Object-Oriented programming languages. Different from points-to based analysis, we do not require a heap abstraction, so the computation is purely relational. We have proved that in the context-insensitive setting, our result is equivalent to that would be produced by a subset-based points-to analysis, regarding the core Object-Oriented language features. We have implemented the algorithm in the Soot compiler framework, and have conducted preliminary evaluation by comparing our results with those produced by the built-in CHA and points-to analysis algorithms in Soot on a selection of $13$ benchmark programs from SPECjvm2008 and DaCapo benchmark suites, and achieved promising results. In the future we plan to develop context-sensitive analysis algorithms based on TFA.

\bibliographystyle{plain}
\bibliography{literature}
\end{document}